\documentclass[11pt]{article}
\usepackage[T1]{fontenc}
\usepackage[utf8]{inputenc}
\usepackage{lmodern}
\usepackage[margin=1in]{geometry}
\usepackage{amsmath,amssymb,amsthm}
\usepackage{booktabs,longtable,array}
\usepackage{graphicx,xcolor}
\usepackage[hidelinks]{hyperref}
\usepackage{microtype}
\newtheorem{theorem}{Theorem}[section]
\newtheorem{lemma}[theorem]{Lemma}
\newtheorem{proposition}[theorem]{Proposition}
\theoremstyle{remark}

\newcommand{\E}{\mathbb{E}}
\newcommand{\Prob}{\mathbb{P}}
\newcommand{\ind}{\mathbf{1}}
\newcommand{\Pois}{\operatorname{Poisson}}
\newcommand{\Bin}{\operatorname{Binomial}}
\newcommand{\Var}{\operatorname{Var}}
\newcommand{\en}{\mathcal{E}}
\newcommand{\code}[1]{\texttt{\small #1}}
\title{An upper bound of 4.268 for the random 3-SAT satisfiability threshold}
\author{Fedor Vorobyev\\\small Independent researcher\\
  \small\href{mailto:fedorvorobev@gmail.com}{fedorvorobev@gmail.com}}
\date{October 2026}
\hypersetup{pdftitle={An upper bound of 4.268 for the random 3-SAT satisfiability threshold},
  pdfauthor={Fedor Vorobyev}}
\begin{document}
\maketitle
\begin{abstract}
We prove that a random 3-SAT formula with $n$ variables and
$\lfloor 4.268n\rfloor$ clauses is unsatisfiable with high probability.
The satisfiability threshold has attracted considerable attention: successive
upper bounds reached 4.4898 in the work of D\'iaz, Kirousis, Mitsche, and
P\'erez-Gim\'enez, while the cavity method predicts a value near 4.267.
Our bound brings the rigorous upper estimate close to this prediction.

The proof builds on earlier interpolation methods, including the energetic
approach of Achlioptas and Menchaca-M\'endez.
We study the smallest number of clauses that any assignment must violate,
comparing the random formula with a simpler system of constraints on
individual variables. This comparison reduces the bound to a finite
calculation, which we verify using exact rational arithmetic.
The proof is formalized in Lean 4.
\end{abstract}

\section{Introduction}
A random 3-SAT formula is formed by sampling clauses independently.
For each clause, choose three distinct variables uniformly and negate
each independently with probability $1/2$. Repeated clauses are allowed.
As the number of clauses increases, the probability that the formula is
satisfiable decreases. We study where this transition occurs, measured
by the ratio of clauses to variables.

Friedgut~\cite{Friedgut} proved that this transition is \emph{sharp}.
As the number of variables grows, the probability of satisfiability drops
from near one to near zero over an interval of clause-to-variable ratios
whose width tends to zero. For example, the ratios at which this
probability falls below $90\%$ and $10\%$ become arbitrarily close.

Sharpness describes the width of the transition, but does not by itself
show that its location settles at a fixed value. The transition may be
centred at a ratio $a_n$ that depends on the number $n$ of variables.
Whether $a_n$ converges to a constant is still open for random 3-SAT.
The cavity prediction is that such a limit exists and is approximately
$4.267$~\cite{MMZ}. Ding, Sly, and Sun~\cite{DSS} proved the existence
and the predicted value of a limiting threshold for random $k$-SAT when
$k$ is sufficiently large; their result does not cover $k=3$.

We bound the location of the transition from above, without assuming
that it converges. Throughout the paper, ``with high probability'' means
that the probability tends to one as the number of variables grows.

\begin{theorem}\label{thm:main}
A random 3-SAT formula in this model with $n$ variables and $\lfloor4.268n\rfloor$
clauses is unsatisfiable with high probability as $n\to\infty$.
\end{theorem}

The proof gives an explicit inequality whose verification uses exact
rational arithmetic. The argument, including this calculation, is
formalized in Lean 4.

\paragraph{Bounds on the threshold.}
A lower bound identifies ratios below which a random formula is
satisfiable with high probability; an upper bound identifies ratios
above which it is unsatisfiable with high probability. These bounds make
sense even without a limiting threshold constant.

Figure~\ref{fig:history} shows selected results from both directions.
The early upper bound of about $5.191$ comes from counting satisfying
assignments~\cite{FP}. Later arguments lowered it to $4.506$~\cite{DBM}
and then $4.4898$~\cite{DKMP}. From below, analyses of algorithms for
finding satisfying assignments reached $3.52$~\cite{HS,KKL06}.
The cavity method predicts a value near $4.267$~\cite{MMZ}.
Our upper bound of $4.268$ comes close to that prediction; it does not
establish the predicted value or close the gap to the lower bound.

\begin{figure}[tbp]
\centering
\input{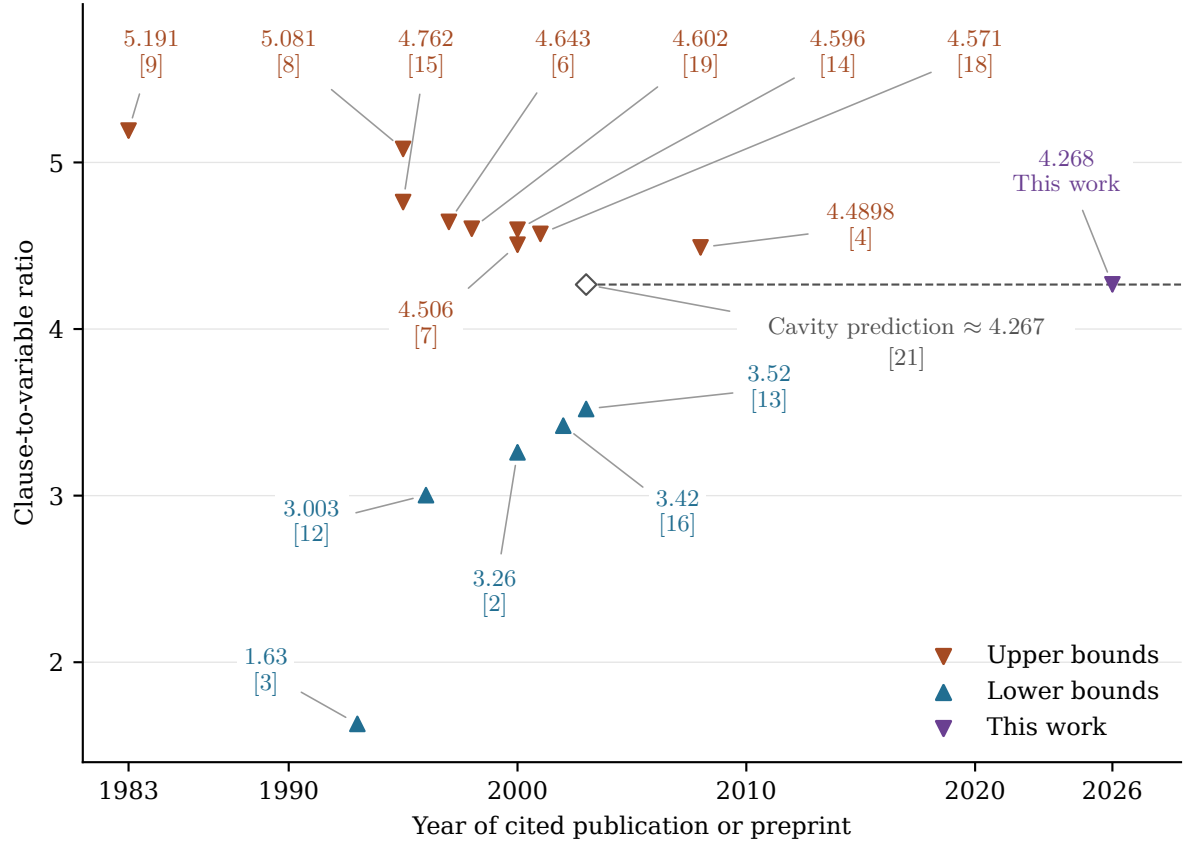}
\caption{Selected bounds on the random 3-SAT satisfiability threshold.
Downward triangles denote upper bounds and upward triangles lower bounds;
the purple downward triangle is Theorem~\ref{thm:main}.
The diamond and dashed line denote the approximate cavity prediction,
not a proved bound. Dates refer to the cited publications or preprints,
not necessarily the first announcements; the prediction is placed at the
2003 first version of~\cite{MMZ}. The $3.52$ point dates the report
of Hajiaghayi and Sorkin~\cite{HS}, which leaves rigorous justification of
its numerical calculations to future work; Kaporis, Kirousis, and
Lalas~\cite{KKL06} independently obtained the same bound.
These are selected results, not a sequence of record improvements.}
\label{fig:history}
\end{figure}

\paragraph{Approach.}
Interpolation methods compare a random formula with a simpler system
by replacing clauses with constraints on individual variables. Our proof
belongs to the line of work developed by Franz and Leone~\cite{FL},
Panchenko and Talagrand~\cite{PT}, and Lelarge and Oulamara~\cite{LO}.
These methods also supply the upper-bound part of the large-$k$ result
of Ding, Sly, and Sun~\cite{DSS}.

A particularly close predecessor is the energetic interpolation of
Achlioptas and Menchaca-M\'endez~\cite{AMM}. They work directly with the
minimum number of violated constraints and compare it with the minimum
for a system of one-variable constraints. Their trial model is controlled
by a single probability. Here we use a distribution of warning rates
and an exponential weight for the minimum violation count. We average
over auxiliary warnings before taking a logarithm. The resulting bound
can be verified by a finite calculation with exact rational arithmetic.
The proof below establishes this comparison directly and formalizes it
in Lean.

\paragraph{Proof outline.}
For a formula $F$, consider the number of clauses violated by each truth
assignment. We denote the smallest of these numbers by $\en(F)$.
Thus $\en(F)=0$ exactly when $F$ is satisfiable. We will show that, for
a random formula at the stated density, this minimum is positive with
high probability.

Set $\alpha=4.268=1067/250$. We first work with a simpler random model:
variables may repeat within a clause, and the number of clauses has a
Poisson distribution with mean $\alpha n$. We compare these formulas
with a system of random constraints, each involving just one variable.
The comparison gives an inequality for the expected minimum number of
violated clauses. An explicit calculation, verified using exact rational
arithmetic, shows that this expectation grows at least linearly with $n$.

We then return to a fixed number $\lfloor\alpha n\rfloor$ of clauses.
The expected minimum remains at least $cn$ for some constant $c>0$.
Changing one clause changes the minimum by at most one, which lets us
bound its variance by $\alpha n$. Chebyshev's inequality then shows
that the probability of a zero minimum is $O(1/n)$. Finally, conditioning
on every clause having three distinct variables increases this bound
by at most a constant factor.

\clearpage
\section{Counting violated constraints and the comparison model}
\subsection{Adding one constraint}
Let $H(\sigma)$ count the constraints violated by a truth assignment
$\sigma\in\{0,1\}^n$, and let $\en(H)=\min_\sigma H(\sigma)$.
We call $H(\sigma)$ its \emph{cost}, and an assignment attaining the
minimum a \emph{minimizer}. For a formula $F$, this gives the quantity
$\en(H)=\en(F)$ from the introduction.

Adding one constraint increases the minimum by zero or one. We need
the exact criterion: the increase is one exactly when every old minimizer violates
the new constraint. If $P$ is the set of assignments violating it and
$\ind_P$ is its indicator (one on $P$, zero elsewhere), then
\begin{equation}\label{eq:integer}
 \en(H+\ind_P)=\en(H)+
 \ind\{\text{every minimizer of }H\text{ belongs to }P\}.
\end{equation}
Any non-minimizing assignment already violates at least one more constraint,
which gives the identity.

For the comparison below we use the weight $e^{-y\en(H)}$, which decreases
with the minimum violation count. Here $y>0$ will be chosen later;
its relation $y=m\beta$ to the partition-function parameters is explained
in Section~\ref{sec:partition}. Write $B_P$ for the indicator on the
right of~\eqref{eq:integer} and $q=1-e^{-y}$. Then
\begin{equation}\label{eq:expinsert}
 e^{-y\en(H+\ind_P)}=e^{-y\en(H)}(1-qB_P).
\end{equation}
\subsection{Replacing clauses with one-variable constraints}
Imagine replacing the second and third literals of a clause by independent
coin tosses. If both come up ``false,'' the first literal must be true to
avoid a violation. Otherwise, there is no constraint. We call this random
one-variable constraint a \emph{site}, following Lelarge and
Oulamara~\cite[Section~4]{LO}. Such one-variable penalties also appear as
univariate factors in Achlioptas and Menchaca-M\'endez~\cite[Section~4]{AMM}.
Our comparison gradually replaces clauses by sites. Once only sites remain,
each variable can be chosen separately to minimize the number of violations.

Here is the precise construction. Start with a random clause on $n\geq1$
variables. Choose its three variable indices $i_1,i_2,i_3$ independently
and uniformly from $\{1,\ldots,n\}$, and choose each literal's sign by a
fair coin toss. Write $J_j$ for the value of variable $i_j$ that makes
literal $j$ false: $J_j=0$ for a positive literal and $J_j=1$ for a
negative literal. The clause's violation indicator is
\[
 C_x(\sigma)=\ind\{\sigma_{i_1}=J_1\}
             \ind\{\sigma_{i_2}=J_2\}\ind\{\sigma_{i_3}=J_3\}.
\]
In this auxiliary model the indices may repeat; we return to clauses
with three distinct variables in Section~\ref{sec:assembly}.

The replacement tosses need not be fair. For each position $j$, independently
choose a probability $R_j$ from the distribution
\begin{equation}\label{eq:law}
 \Prob(R=1/20)=1/4,\qquad
 \Prob(R=2/5)=1/2,\qquad
 \Prob(R=17/20)=1/4.
\end{equation}
Given these probabilities, toss two independent coins, and let $W_j=1$
mean that coin $j$ comes up ``false,'' so
$\Prob(W_j=1\mid R_j)=R_j$ for $j=2,3$.
Replace the last two indicators in $C_x$ by $W_2,W_3$:
\[
 S_{x,W}(\sigma)=\ind\{\sigma_{i_1}=J_1\}W_2W_3.
\]
Thus the site penalizes a false first literal with probability $R_2R_3$,
and otherwise contributes zero. We refer to $W_2,W_3$ as auxiliary
\emph{warnings}. Their probabilities are chosen for the comparison;
they are not inferred from the original formula. The particular
distribution~\eqref{eq:law} gives the numerical bound verified in
Section~\ref{sec:certificate}.

For the calculations below, write $x_j=(i_j,J_j,R_j)$ and call this
triple a \emph{seed}; call the three seeds together a \emph{row}.
Thus a row records a random clause and the probabilities for its replacement.
The clauses themselves do not use the $R_j$, and the replacement above
does not use $R_1$. Different rows are independent.

\subsection{The quantity we compare}
Consider $K$ clauses and $L$ sites, with all literals and coin probabilities
already chosen. Tossing the coins determines which sites impose a constraint.
For each outcome, we find the minimum number of violations over all assignments.
Write $a$ for the collection of coin outcomes and $H_a(\sigma)$ for the
number of constraints violated by assignment $\sigma$. The minimum is
$E_a=\min_\sigma H_a(\sigma)$.

We first average the weight $e^{-yE_a}$ over the coin tosses, keeping the
literals and coin probabilities fixed. Call this average $Z_y$, and its
logarithm $\Phi_y$. We then average $\Phi_y$ over the choices of literals
and coin probabilities---the rows defined above---to obtain $f(K,L)$:
\begin{equation}\label{eq:functional}
 Z_y=\E_a e^{-yE_a},\qquad \Phi_y=\log Z_y,\qquad
 f(K,L)=\E_{\rm rows}\Phi_y.
\end{equation}
Here $\E_a$ averages only over the independent coin tosses with their
chosen probabilities; $\E_{\rm rows}$ averages over the rows.
The order matters: we take the logarithm between the two averages.
The quantity $f(K,L)$ lets us compare systems with different numbers
of clauses and sites.

We will use two simple bounds. At most $K+L$ constraints can be violated,
so $0\leq E_a\leq K+L$ and hence
\begin{equation}\label{eq:growth}
 -y(K+L)\leq\Phi_y\leq0.
\end{equation}
Also, adding one clause or site increases the minimum by at most one
for every outcome of the tosses. The weight $e^{-yE_a}$ therefore
decreases by a factor between $e^{-y}$ and $1$, as does its average $Z_y$.
Taking logarithms, the change in $\Phi_y$ lies in $[-y,0]$.
The same bound holds after averaging over the rows, and thus for $f$.
All the averages in this section are finite sums.

\section{The local cubic inequality}
We compare the effects of adding a clause and adding a site so that we
can replace the original formula by a simpler system while keeping a
bound on its minimum number of violations. Once only sites remain,
this bound can be computed explicitly. The key comparison follows from
a cubic polynomial that factors as $(B-C)^2(B+2C)\geq0$ for $B,C\geq0$.
This inequality controls how our comparison quantity changes during
the replacement.

Keep the literals and coin probabilities fixed. We now reweight the
possible coin outcomes: multiply the probability of each outcome $a$ by
$e^{-yE_a}$, then divide by $Z_y$ so that the probabilities sum to one.
This defines
\[
 \mu_y(a)=\frac{\Prob(a)e^{-yE_a}}{Z_y}.
\]
Here $\Prob(a)$ is the probability of the coin outcomes with the rows
held fixed. Outcomes with fewer unavoidable violations receive greater
weight under $\mu_y$. Let $\langle\cdot\rangle$ denote expectation with
respect to this reweighted distribution.

For a new independent seed $x=(i,J,R)$ define
\[
 b(x,a)=\ind\{\sigma_i=J\text{ for every minimizer of }H_a\}.
\]
Thus $b(x,a)$ is one exactly when the literal is false in every
assignment attaining the minimum violation count.
Write $\Delta_C\Phi_y$ and $\Delta_S\Phi_y$ for the changes caused by
adding a clause or a site from a new independent row $(x_1,x_2,x_3)$:
in each case, the value after addition minus the value before.
We use $\E_x$ for an average over one seed and $\E_{x_1,x_2,x_3}$
for an average over the three independent seeds of this new row.

A clause is violated by every minimizer exactly when each of its three
literals is falsified by every minimizer. This is the elementary
commutation of the two universal quantifiers over minimizers and literal
positions; it remains valid when variable positions repeat. Consequently,
\eqref{eq:expinsert} gives the clause increment
\begin{equation}\label{eq:clause}
 \Delta_C\Phi_y=\log\bigl(1-q\langle b(x_1,a)b(x_2,a)b(x_3,a)\rangle\bigr).
\end{equation}
For a site, first average its two new coin tosses, which are independent
of all existing coin outcomes. Their product has mean $R_2R_3$, hence
\begin{equation}\label{eq:site}
 \Delta_S\Phi_y=\log\bigl(1-q\langle b(x_1,a)R_2R_3\rangle\bigr).
\end{equation}
Define the correction
\begin{equation}\label{eq:dy}
 d_y=\E\log(1-qR_1R_2R_3),
\end{equation}
where the rates are independent copies of~\eqref{eq:law}.

\begin{lemma}\label{lem:cubic}
Averaging over the new independent row gives
\begin{equation}\label{eq:local}
 \E_{x_1,x_2,x_3}\Delta_C\Phi_y-3\E_{x_1,x_2,x_3}\Delta_S\Phi_y+2d_y\leq0.
\end{equation}
\end{lemma}
\begin{proof}
For $0\leq z\leq1$,
$\log(1-qz)=-\sum_{h=1}^{\infty}q^hz^h/h$. The absolute sum is bounded by
$\sum_{h=1}^{\infty}q^h/h=y$, uniformly in $z$, so it may be averaged termwise.
For each positive integer $h$, corresponding to the $h$-th term of this
series, sample $h$ independent copies $a_1,\ldots,a_h$ of the coin
outcomes from the distribution $\mu_y$, and set
\[
 B_h=\E_x\prod_{t=1}^h b(x,a_t),\qquad C_h=\E R^h.
\]
The seed $x$ is shared across the $h$ factors in $B_h$. Write
$\langle\cdot\rangle_h$ for expectation over the independent copies
$a_1,\ldots,a_h$. Expanding each power using these copies and averaging
the three new seeds independently gives
\begin{align*}
 \E_{x_1,x_2,x_3}\bigl[\langle b(x_1,a)b(x_2,a)b(x_3,a)\rangle^h\bigr]
   &=\langle B_h^3\rangle_h,\\
 \E_{x_1,x_2,x_3}\bigl[\langle b(x_1,a)R_2R_3\rangle^h\bigr]
   &=\langle B_h C_h^2\rangle_h,\\
 \E_{x_1,x_2,x_3}[(R_1R_2R_3)^h]&=C_h^3.
\end{align*}
Thus the three terms in~\eqref{eq:local} combine, at each power $h$,
into the expectation of the following nonnegative expression:
\[
 B_h^3-3B_hC_h^2+2C_h^3=(B_h-C_h)^2(B_h+2C_h)\geq0.
\]
Multiplying by $-q^h/h\leq0$ and summing over $h$ proves the claim.
\end{proof}

Averaging the existing rows in~\eqref{eq:local} yields
\begin{equation}\label{eq:difference}
 f(K+1,L)-3f(K,L+1)+2f(K,L)+2d_y\leq0.
\end{equation}
This factorization is the cubic case of the polynomial remainder
in Lelarge and Oulamara~\cite[condition (4), equation (30)]{LO}.

\section{From clauses to sites with Poisson counts}
We now apply Section~3's inequality repeatedly to compare a system
containing only clauses with one containing only sites. We first make
finitely many additions, then average over their number. Using a Poisson
distribution for this number gives the random clause and site counts
needed below. Here $\Pois(a)$ denotes the Poisson distribution with mean $a$.
The next comparison applies to any function with the stated properties.

\begin{lemma}\label{lem:comparison}
Suppose $f:\mathbb{N}^2\to\mathbb{R}$ has uniformly bounded increments
in each coordinate: increasing either argument by one changes $f$ by
at most a fixed constant in absolute value. Suppose also that it
satisfies~\eqref{eq:difference} with a constant $d$ in place of $d_y$.
For every $a\geq0$,
\begin{equation}\label{eq:comparison}
 \E_{K\sim\Pois(a)}f(K,0)
 \leq\E_{L\sim\Pois(3a)}f(0,L)-2ad.
\end{equation}
\end{lemma}
\begin{proof}
Define
\[
 T_N(K,L)=\sum_{j=0}^N\binom Nj2^{N-j}f(K+j,L).
\]
After division by $3^N$, this is the average of $f(K+J,L)$ for
$J\sim\Bin(N,1/3)$: the number of successes in $N$ independent trials,
each succeeding with probability $1/3$. It therefore averages over a
random number of clause additions.
Pascal's identity gives
$T_{N+1}(K,L)=T_N(K+1,L)+2T_N(K,L)$. Induction implies
\begin{equation}\label{eq:binomial}
 T_N(K,L)\leq3^N\bigl(f(K,L+N)-2Nd/3\bigr).
\end{equation}
Indeed, the case $N=0$ is equality. Applying the induction hypothesis to
both summands of $T_{N+1}$ gives an upper bound
\[
 3^N\bigl(f(K+1,L+N)+2f(K,L+N)-2Nd\bigr).
\]
The assumed local inequality at $(K,L+N)$ bounds this by
$3^{N+1}(f(K,L+N+1)-2(N+1)d/3)$.

At $K=L=0$, divide~\eqref{eq:binomial} by $3^N$ and average
$N\sim\Pois(3a)$. Conditional on $N$, the left side averages
$f(J,0)$ with $J\sim\Bin(N,1/3)$. For $N=j+\ell$, its weight is
\[
 e^{-3a}\frac{(3a)^N}{N!}\binom Nj(1/3)^j(2/3)^\ell
 =e^{-3a}\frac{a^j}{j!}\frac{(2a)^\ell}{\ell!}.
\]
Summing $\ell$ gives the $\Pois(a)$ weight of $j$, whereas
$\E N=3a$. Bounded increments give at most linear growth, justifying
all rearrangements by absolute convergence. The case $a=0$ is immediate.
\end{proof}

Let $F^{\rm rep}_K$ have $K$ clauses whose variable positions are sampled
with replacement, so variables may repeat within a clause; this is what
the superscript ``rep'' denotes. With no sites there are no auxiliary
coin tosses to average over, so $f(K,0)=-y\E\en(F^{\rm rep}_K)$. With $a=\alpha n$,
Lemma~\ref{lem:comparison} becomes
\begin{equation}\label{eq:poissonenergy}
 -\frac yn\E\en(F^{\rm rep}_{\Pois(\alpha n)})
 \leq\frac1n\E_{L\sim\Pois(3\alpha n)}f(0,L)-2\alpha d_y.
\end{equation}
Here a Poisson subscript means that the count is sampled independently
and then the indicated number of independent clauses is sampled.

\section{The site-only system and its moments}
To bound the site-only quantity in Section~4, we calculate the probability
that a variable receives conflicting constraints: some active sites
penalize zero and others penalize one. Such a variable must violate at
least one site, whichever value we assign to it.

\subsection{Conflicting constraints}
Hold the site rows fixed. Let $\mathcal S_{v,s}$ be the set of sites
whose retained variable is $v$ and whose penalized value is $s\in\{0,1\}$.
Let $N_{v,s}$ count the active sites in this set, meaning those for which
both replacement coins come up ``false.''
Every assignment violates at least one site at a variable where both
$N_{v,0}$ and $N_{v,1}$ are positive. Therefore the minimum number of violated sites
is at least the number of such variables.

Let $A_v$ be the probability, over the coin tosses, that no active site
penalizes zero at $v$, and let $B_v$ be the corresponding probability for
one. Let $G_v$ be the probability that at least one of these values is
unpenalized. Then
\[
 A_v=\prod_{\mathcal S_{v,0}}(1-R_2R_3),\quad
 B_v=\prod_{\mathcal S_{v,1}}(1-R_2R_3),\quad
 G_v=A_v+B_v-A_vB_v.
\]
Each factor is the probability that the corresponding site is inactive;
an empty product is one. The two groups use different, independent coin
tosses, which gives the displayed expression for $G_v$.
Different variables also use different, independent tosses. Writing
$\delta=e^{-y}$, the lower bound on violations and this independence imply
\begin{equation}\label{eq:vertex}
 \Phi_y\leq\sum_{v=1}^n\log\bigl(G_v+(1-G_v)\delta\bigr).
\end{equation}
This is an inequality: counting conflicting variables need not give the
exact minimum number of violations.

To obtain a finite calculation with a controlled error, we retain only
$T$ terms of the logarithmic series. The following bound controls the
error from this step.

\begin{lemma}\label{lem:truncate}
For $x\in[0,1]$, $\delta\in(0,1]$, and an integer $T\geq1$,
\begin{equation}\label{eq:truncate}
 \log(1-(1-\delta)x)\leq-\sum_{k=1}^T\frac{x^k}{k}+T\delta.
\end{equation}
\end{lemma}
\begin{proof}
Truncating the logarithmic series gives an upper bound
$-\sum_{k=1}^T(1-\delta)^kx^k/k$. Bernoulli's inequality gives
$(1-\delta)^k\geq1-k\delta$, and $x^k\leq1$ bounds the total correction
by $T\delta$. This also covers $x=1$ and $\delta=1$.
\end{proof}

\subsection{Mixed Poisson moments}
Let $L\sim\Pois(3\alpha n)$ and fix a variable $v$. We need expectations
of powers, called moments, of the probabilities $A_v$ and $B_v$.
Write $M_j=\E A_v^j$ and $\mu_k=\E(1-A_v)^k$, where the expectations
include the random site count and rows. We will show that these moments
can be computed as follows:
\begin{equation}\label{eq:moments}
 \nu_j=\E(1-RR')^j,\qquad
 M_j=\exp\bigl((3\alpha/2)(\nu_j-1)\bigr),\qquad
 \mu_k=\sum_{j=0}^k(-1)^j\binom kj M_j.
\end{equation}
All rate copies in the same expectation are independent. For independent
rows $X_i$ and a function $0\leq h\leq1$, conditioning on $L$ gives
\[
 \E\prod_{i=1}^L h(X_i)
 =\exp\bigl(3\alpha n(\E h(X)-1)\bigr).
\]
A site belongs to each of $\mathcal S_{v,0}$ and $\mathcal S_{v,1}$
with probability $1/(2n)$, and cannot belong to both. Its contribution to $A_v^jB_v^\ell$
therefore has mean
$1+(\nu_j-1)/(2n)+(\nu_\ell-1)/(2n)$. Hence
\begin{equation}\label{eq:mixed}
 \E[A_v^jB_v^\ell]
 =\exp\bigl((3\alpha/2)(\nu_j+\nu_\ell-2)\bigr)=M_jM_\ell.
\end{equation}
Taking $\ell=0$ verifies the formula for $M_j$; expanding $(1-A_v)^k$
verifies the formula for $\mu_k$. These expectations of products of
powers are the mixed moments. Expanding both factors now yields
\begin{equation}\label{eq:musquare}
 \E\bigl[((1-A_v)(1-B_v))^k\bigr]=\mu_k^2.
\end{equation}
In particular $\mu_k=\E(1-A_v)^k\geq0$.

Apply Lemma~\ref{lem:truncate} with $x=(1-A_v)(1-B_v)=1-G_v$,
then average~\eqref{eq:vertex} and use~\eqref{eq:musquare}. We obtain
\begin{equation}\label{eq:sitebound}
 \frac1n\E_{L\sim\Pois(3\alpha n)} f(0,L)
 \leq-\sum_{k=1}^T\frac{\mu_k^2}{k}+Te^{-y}.
\end{equation}
Only a finite polynomial has been integrated. In particular, no
integrability claim for $\log G_v$ or $1/G_v$ is needed.

Since $q\leq1$ and $R_1R_2R_3<1$, equation~\eqref{eq:dy} implies
$d_y\geq\E\log(1-R_1R_2R_3)$. Define
\begin{equation}\label{eq:psi}
\Psi_T=-\sum_{k=1}^T\frac{\mu_k^2}{k}
             -2\alpha\E\log(1-R_1R_2R_3).
\end{equation}
The quantity $\Psi_T$ collects the two parts of our bound: the truncated
sum comes from the system of one-variable constraints, and the second
term is the correction for replacing clauses by those constraints.
It depends only on the density, the chosen distribution of coin
probabilities $R$, and the number $T$ of retained terms. We choose this
distribution to obtain a negative bound.
Combining~\eqref{eq:poissonenergy} and~\eqref{eq:sitebound} proves
\begin{proposition}\label{prop:energy}
For every $n\geq1$, $y>0$, and $T\geq1$,
\begin{equation}\label{eq:energybound}
 -\frac yn\E\en(F^{\rm rep}_{\Pois(\alpha n)})\leq\Psi_T+Te^{-y}.
\end{equation}
\end{proposition}

\section{The rational certificate}\label{sec:certificate}
We now prove that the bound from Section~5 is negative using exact
rational calculations. This finite calculation is our numerical
certificate. For $\alpha=4.268$, the distribution~\eqref{eq:law}, and
$T=128$, the required inequality is
\begin{equation}\label{eq:negative}
 \Psi_{128}\leq-\frac1{100000}=-\varepsilon.
\end{equation}
We specify a rational calculation proving this inequality. The displayed
decimal approximations below are explanatory only.

Represent the rate law by four equally likely labels
\[
 (r_0,r_1,r_2,r_3)=(1/20,17/20,2/5,2/5).
\]
Then $\nu_j=\frac1{16}\sum_{i,\ell=0}^3(1-r_ir_\ell)^j$ is rational.
Put $x_j=(3\alpha/2)(\nu_j-1)$ for $0\leq j\leq128$.

\subsection{Exponential enclosures}
An enclosure is an interval guaranteed to contain the true value.
We round its lower endpoint down and its upper endpoint up to multiples
of $1/s$, where $s$ is a positive integer. This outward rounding keeps
the true value inside the interval. The rounding functions are
$D_s(z)=\lfloor sz\rfloor/s$ and $U_s(z)=-D_s(-z)$.
Define intervals recursively by
\begin{align*}
 I_0(x)&=[\max(0,1+x),\,1/(1-x)],\\
 I_{t+1}(x)&=[D_s(l^2),U_s(u^2)],
       \qquad [l,u]=I_t(x/2).
\end{align*}
For $x<2^t$, the interval $I_t(x)$ encloses $e^x$ and has nonnegative
endpoints. At $t=0$, use $1+x\leq e^x$ and
$1-x\leq e^{-x}$; the latter can be inverted since $x<1$.
Induction uses $e^x=(e^{x/2})^2$ and outward rounding.

Compute $I_{160}(x_j)$ on the grid $s=2^{320}$, then round its endpoints
outward on the grid $2^{160}$ to obtain $[l_j,u_j]$. The domain condition
holds because $x_j\leq0$. For $c_{kj}=(-1)^j\binom kj$, put
\[
 L_k=\sum_{j=0}^k c_{kj}
 \begin{cases}l_j,&c_{kj}\geq0,\\u_j,&c_{kj}<0,\end{cases}
 \qquad V=-\sum_{k=1}^{128}\frac{\max(0,L_k)^2}{k}.
\]
Using lower endpoints for positive coefficients and upper endpoints
for negative coefficients gives $L_k\leq\mu_k$. Since $\mu_k\geq0$,
$\max(0,L_k)^2\leq\mu_k^2$, so $V$ bounds the site contribution from above.
We replace a negative lower bound by zero before squaring; squaring
the negative bound itself would not give a valid lower bound on $\mu_k^2$.

\subsection{Logarithmic enclosures and exact arithmetic}
For $0\leq z<1$, let
\[
 \lambda(z)=-\sum_{h=1}^{32}\frac{z^h}{h}-\frac{z^{33}}{1-z}.
\]
Bounding the remaining logarithmic-series tail by a geometric series
proves $\lambda(z)\leq\log(1-z)$. Thus
\[
 W=-\frac{2\alpha}{64}\sum_{i,j,\ell=0}^3
                   \lambda(r_ir_jr_\ell)
\]
is an upper bound on the clause-replacement correction. The resulting entirely
rational inequality is
\begin{equation}\label{eq:rational}
 V+W\leq-1/100000.
\end{equation}
The rational inequality~\eqref{eq:rational} is verified in Lean;
the analytic enclosure lemmas then imply~\eqref{eq:negative}.
Technical verification details appear in Appendix~\ref{app:verification}.
The exponential intervals are computed within Lean from their recurrence;
no table supplied by an external program is a premise of the theorem.

An independent Python implementation uses rational Taylor enclosures
for the exponential and a different logarithm enclosure. Its upper
bounds, rounded for display, are
\begin{center}
\begin{tabular}{lr}
\toprule
Contribution & Independent upper bound \\
\midrule
Site contribution & $-0.7611927541550103$\\
Clause-replacement correction & $+0.7611815516097719$\\
Sum & $-0.000011202545238429759$\\
\bottomrule
\end{tabular}
\end{center}
These numbers cross-check the calculation; they are not inputs to Lean
and are not claimed to be exact values of the real expression $\Psi_{128}$.

\section{Fixed clause count and three distinct variables per clause}\label{sec:assembly}
We now return from the auxiliary model to a fixed number of clauses,
each containing three distinct variables, and prove unsatisfiability
with high probability. Let $F_n$ denote the random formula in
Theorem~\ref{thm:main}.
Choose one finite $y>0$ such that $128e^{-y}\leq\varepsilon/4$, and keep
it fixed as $n$ grows. For example, $y=\log(512\cdot10^5)$ gives
$128e^{-y}=\varepsilon/4$ exactly. Let $M=\lfloor\alpha n\rfloor$ and
$N\sim\Pois(\alpha n)$. Generate one infinite sequence of independent
clauses and compare the formulas formed by its first $N$ and first $M$
clauses. Each additional clause changes the minimum violation count by
at most one, so
\begin{align}
 \bigl|\E\en(F^{\rm rep}_N)-\E\en(F^{\rm rep}_M)\bigr|
 &\leq\E|N-M|\nonumber\\
 &\leq\sqrt{\E(N-M)^2}\leq\sqrt{\alpha n+1}.\label{eq:depois}
\end{align}
Here $\E(N-M)^2=\alpha n+(\alpha n-M)^2$. Equations
\eqref{eq:energybound}, \eqref{eq:negative}, and~\eqref{eq:depois} yield
\[
 -\frac yn\E\en(F^{\rm rep}_M)
 \leq-\varepsilon+128e^{-y}+\frac yn\sqrt{\alpha n+1}.
\]
Eventually the last term is at most $\varepsilon/4$, so, for
$c=\varepsilon/(2y)>0$,
\begin{equation}\label{eq:linear}
 \E\en(F^{\rm rep}_M)\geq cn
\end{equation}
for all sufficiently large $n$. This fixes $y$ before taking $n\to\infty$;
there is no exchange of two limits.

A positive expected minimum alone does not establish unsatisfiability
with high probability. We must also bound fluctuations around that mean.
Replacing one clause changes every assignment's violation count by at
most one, hence changes its minimum by at most one. Reveal the $M$ clauses
one at a time. After the first $j$ clauses are revealed, take the
conditional expectation of the final minimum $\en(F^{\rm rep}_M)$.
This sequence of conditional expectations is the Doob martingale. Conditional on the preceding clauses, the
conditional expectations corresponding to different values of the next
clause differ by at most one: use the same unrevealed clauses in both cases.
Each martingale increment therefore has absolute value at most one.
Each increment has conditional mean zero. Distinct increments have zero
covariance, so their variances add, giving
\begin{equation}\label{eq:variance}
 \Var(\en(F^{\rm rep}_M))\leq M\leq\alpha n.
\end{equation}
The formal proof uses the equivalent finite-product variance decomposition.
Chebyshev's inequality and~\eqref{eq:linear} imply eventually
\begin{equation}\label{eq:chebyshev}
 \Prob(\en(F^{\rm rep}_M)=0)\leq\frac{\alpha}{c^2n}.
\end{equation}

Let $A$ be the event that each of the $M$ clauses has distinct variable
positions. Conditional on $A$, each clause has three uniformly chosen
distinct variables, and clauses and signs remain independent, with fair
signs. This is exactly the original model. Conditioning increases our
probability bound by at most the constant below: for $n\geq6$,
\begin{equation}\label{eq:conditioning}
 \Prob(A)^{-1}=
 \left(\frac{n^2}{(n-1)(n-2)}\right)^M
 \leq e^{6M/n}\leq e^{6\alpha}.
\end{equation}
Indeed, for $r=n^2/((n-1)(n-2))$, use
$\log r\leq r-1\leq6/n$; the latter inequality reduces to
$3n^2-16n+12\geq0$, valid for $n\geq6$.
Since $\Prob(B\mid A)\leq\Prob(B)/\Prob(A)$ for any event $B$,
\eqref{eq:chebyshev} and~\eqref{eq:conditioning} give
\[
 \Prob(F_n\text{ is satisfiable})
 \leq e^{6\alpha}\frac{\alpha}{c^2n}\longrightarrow0.
\]
This proves Theorem~\ref{thm:main}. The conditioning constant is deliberately
coarse; the argument supplies an asymptotic conclusion, not a useful
small-instance estimate.

The conclusion also holds when repeated clauses are forbidden.
There are $N_n=8\binom n3$ distinct signed clauses, ignoring the order
of their literals. Among $M=\lfloor\alpha n\rfloor$ independent samples,
the probability of any repetition is at most $\binom M2/N_n=O(1/n)$.
Conditioning on no repetitions gives uniform sampling without replacement,
and divides the satisfiability bound by a probability tending to one.

\section{Relation to the partition-moment formulation}\label{sec:partition}
This section connects our proof to earlier interpolation formulas;
it is not needed to prove the main theorem.
Keep the same cost function $H_a$ and the same probabilities for the
auxiliary coin outcomes. For fixed outcomes $a$, define the
\emph{partition function}
\[
 Z_\beta(a)=\sum_\sigma e^{-\beta H_a(\sigma)},\qquad \beta\geq0.
\]
The parameter $\beta$ penalizes violations: larger values give less
weight to assignments with more violations. Unlike $Z_y$, which averages
a weight of the minimum over coin outcomes, $Z_\beta(a)$ sums weights
over all assignments for fixed coin outcomes. There are $2^n$ assignments,
and at least one attains the minimum, so
\[
 e^{-\beta E_a}\leq Z_\beta(a)\leq2^n e^{-\beta E_a}.
\]
Introduce an exponent $m>0$. Raise the bounds to this power, average
over the auxiliary coin outcomes, and take logarithms (with $\beta>0$):
\begin{equation}\label{eq:bridge}
 \Phi_{m\beta}\leq\log\E_a[Z_\beta(a)^m]
                  \leq\Phi_{m\beta}+mn\log2.
\end{equation}
Thus at fixed $m\beta=y$, the difference between the middle expression
and $\Phi_y$, divided by $n$, is at most $m\log2$.
This two-sided bound, also formalized in Lean, relates our comparison
quantity to the logarithm of an averaged power of a partition function,
as in
Franz--Leone~\cite[equation (45)]{FL}.

The comparison in Section~4 retains the clause/site balance of their
Poisson interpolation~\cite[equations (51)--(54)]{FL}, but proves it
using finite replacement. Equation~\eqref{eq:bridge} uses an identical
distribution of auxiliary coin outcomes on its two sides; it does not
identify that distribution with the reweighted measure in their
equation (46). Nor do we claim the exact
integer-cost implementation appears in that paper. The direct insertion
identity is sufficient here, at the cost of restricting the argument to
integer unit costs rather than a general finite-temperature model.
For clauses with an odd number of literals, the cited Franz--Leone
preprint makes additional assumptions about overlaps, quantities
measuring correlations between assignments. Our comparison uses the
explicitly nonnegative quantities $B_h,C_h$ from Section~3, for which
the cubic inequality holds directly without such assumptions.

\section{LLM contribution and formal verification}\label{sec:lean}
\paragraph{Use of language models.}
Large language models drove the research leading to this paper. Their
contribution included developing the mathematical argument, generating
and revising the Lean formalization and numerical code, and drafting
and revising the manuscript. The author directed the work and worked
with the models to improve the explanations and presentation. The author
takes responsibility for the mathematical claims, attribution, and final text.

\paragraph{Formal verification.}
The proof of Theorem~\ref{thm:main} is formalized in Lean 4 using
mathlib~\cite{mathlib}. This includes the comparison inequality, the
exact rational certificate, and the transfer to a fixed clause count
with three distinct variables per clause. Lean checks the proof of the
formal theorem; the correspondence between the formal definitions and
the intended model, and between the formal proof and this exposition,
requires human inspection. We do not claim an independent human
verification of every step.

The formal sources and numerical code are available in the repository
and versioned archive below. Appendix~\ref{app:verification} records
the formal theorem statement, its correspondence with the paper,
the verification environment, and reproduction instructions.
\paragraph{Availability.}
The Lean sources, independent certificate code, and reproduction instructions
are available at
\url{https://github.com/yamelton/sat-upper-4268}.
The versioned software archive for release \texttt{v1.0.0} is preserved at
\url{https://doi.org/10.5281/zenodo.23120101}.
The software is licensed under Apache-2.0, with
third-party notices preserved. The manuscript is excluded from the software
license; copyright in the manuscript is retained by Fedor Vorobyev.

\paragraph{Acknowledgments.}
The author thanks the Lean and mathlib communities for the proof
assistant and mathematical library used in this work. The mathematical
debt to earlier interpolation methods is described in the introduction
and Section~\ref{sec:partition}.

\bibliographystyle{plain}
\bibliography{references}

\clearpage
\appendix
\section{Formalization and reproduction details}\label{app:verification}
\subsection{Formal statement and correspondence with the paper}
The model expresses probability as the number of satisfiable formulas
divided by the total number of formulas. Each clause is represented by
three distinct variable indices and three Boolean values specifying when
its literals are false. A formula is an ordered list of independently
sampled clauses. The final theorem statement is
\begin{verbatim}
Filter.Tendsto
  (fun n : Nat => SatUpper.satisfiabilityProbability n (1067*n/250))
  Filter.atTop (nhds 0)
\end{verbatim}
Natural-number division gives exactly $\lfloor1067n/250\rfloor$.
For $n<3$, there are no clauses with three distinct variables, and the
probability definition uses Lean's division-by-zero convention. These
finitely many indices do not affect the limit. The interpolation,
numerical, and concentration arguments are proved in the development;
they are not assumptions of the final theorem.

The following table identifies the modules corresponding to the main
steps of the paper.

\begin{center}
\small
\begin{tabular}{>{\raggedright\arraybackslash}p{.42\textwidth}>{\raggedright\arraybackslash}p{.50\textwidth}}
\toprule
Mathematical step & Modules below \code{lean/SatUpper/}\\
\midrule
Model and final theorem & \code{Model}, \code{Threshold}\\
Minimum violation count & \code{Ground/Minimum}, \code{Ground/Finite}\\
Insertion and cubic sign & \code{Ground/Insertion}, \code{LogInsertion}, \code{ReplicaCubic}\\
Finite replacement and Poisson average & \code{BinomialPoisson}, \code{PoissonExpectation}, \code{Ground/Comparison}\\
Site bound and mixed moments & \code{Ground/Site}, \code{Ground/SiteBound}, \code{PoissonPair}\\
Rational certificate & \code{GridExponential}, \code{NumericalMoments}, \code{NumericalLogs}, \code{NumericalEnclosure}\\
Fixed count and asymptotics & \code{ReplacementPoissonMean}, \code{ModelTransfer}, \code{Obligations}\\
Variance and conditioning & \code{FiniteVariance}, \code{VarianceControl}, \code{LegalConditioning}\\
Partition-moment bound & \code{Ground/Minimum}, \code{Ground/Finite}\\
\bottomrule
\end{tabular}
\end{center}

The audit asks Lean which axioms are used anywhere in each project
theorem's chain of dependencies and separately checks the final theorem
statement above. It permits only \code{propext}, \code{Classical.choice},
and \code{Quot.sound}, the standard axioms for propositional
extensionality, choice, and quotients. It rejects proof placeholders and
additional evaluation axioms. The audit inspects Lean's environment;
the proofs themselves are checked by Lean's kernel.

The Python calculation is a separate implementation of the numerical
bound; its output is not a premise of the Lean proof.
Lean proves the finite rational inequality~\eqref{eq:rational} using
\code{decide +kernel}, which performs computation checked by Lean's kernel.

\newpage
\subsection{Reproduction}
The formalization uses Lean 4.24.0 and mathlib~\cite{mathlib}, pinned to
revision
\begin{center}
\code{f897ebcf72cd16f89ab4577d0c826cd14afaafc7}.
\end{center}
With the pinned toolchain installed, the proof and audit are reproduced
from the repository root by
\begin{verbatim}
sh lean/run.sh exe cache get
python3 lean/verify.py
python3 lean/test_verify.py
\end{verbatim}
The independent certificate uses only the Python standard library:
\begin{verbatim}
python3 src/richer_certificate.py --alpha 4.268 --cutoff 128 \
  --surveys .05,.85,.5 .4,.4,.5 --output certificate.json
\end{verbatim}
The output field \code{psi\_upper} gives a rational upper bound on
$\Psi_{128}$; its \code{exact} entry is the fraction used for checking,
while \code{decimal\_for\_display} is only an approximation.
The flag \code{certified\_negative} records whether this rational upper
bound is negative.
The repository additionally includes numerical cross-checks, source
provenance, and instructions for rebuilding the paper and release bundle.
\end{document}